\documentclass[conference]{IEEEtran}

\usepackage{epsfig}
\usepackage{psfrag}
\usepackage{amsmath}
\usepackage{amsfonts}
\usepackage{amssymb}
\usepackage{subfigure}
\usepackage{booktabs}
\usepackage{macros}
\usepackage{flushend}

\newtheorem{theorem}{Theorem}
\newtheorem{lemma}[theorem]{Lemma}

\newtheorem{note}{Note}

\newcommand{\beq}{\begin{equation}}
\newcommand{\eeq}{\end{equation}}

\newcommand{\beqn}{\begin{eqnarray}}
\newcommand{\beqqn}{\begin{eqnarray*}}
\newcommand{\eeqn}{\end{eqnarray}}
\newcommand{\eeqqn}{\end{eqnarray*}} 
\newcommand{\ba}{\begin{array}}
\newcommand{\ea}{\end{array}}

\newcommand{\sqbullet}{\scriptscriptstyle\blacksquare}

\title{The Capacity Loss of Dense Constellations}

\author{
\authorblockN{Tobias Koch}
\authorblockA{University of Cambridge\\
\tt tobi.koch@eng.cam.ac.uk}
\and
\authorblockN{Alfonso Martinez}
\authorblockA{
Universitat Pompeu Fabra\\
\tt alfonso.martinez@ieee.org}
\and
\authorblockN{Albert Guill{\'e}n i F{\`a}bregas}
\authorblockA{
ICREA \& Universitat Pompeu Fabra\\
University of Cambridge\\
\tt guillen@ieee.org}
}

\begin{document}

\maketitle

\begin{abstract}
We determine the loss in capacity incurred by using signal constellations with a bounded support over general complex-valued additive-noise channels for suitably high signal-to-noise ratio. Our expression for the capacity loss recovers the power loss of 1.53dB for square signal constellations.
\renewcommand{\thefootnote}{}
  \footnote{The research leading to these results has received funding
    from the European Community's Seventh Framework Programme
    (FP7/2007-2013) under grant agreement No. 252663 and from the European Research Council under ERC grant agreement 259663.}
\end{abstract}
\setcounter{footnote}{0}

\section{Introduction}

As it is well known, the channel capacity of the complex-valued Gaussian channel with input power at most $\const{P}$ and noise variance $\sigma^2$ is given by \cite{shannon48}
\beq
C_{\textnormal{G}}(\const{P},\sigma) = \log\left(1+\frac{\const{P}}{\sigma^2}\right).
\eeq
Although inputs distributed according to the Gaussian distribution attain the capacity, they suffer from several drawbacks which prevent them from being used in practical systems. Among them, especially relevant are the unbounded support and the infinite number of bits needed to represent signal points. 

In practice, discrete distributions with a bounded support are typically preferred---in this case, the number of points is allowed to grow with the signal-to-noise ratio (SNR). Ungerboeck computed the rates that are achievable over the Gaussian channel when the channel input takes value in a finite constellation \cite{ungerboeck82}. He observed that, when transmitting at a rate of $R$ bits per channel use, there is not much to be gained from using constellations with size $\const{N}$ larger than $2^{R+1}$. Ozarow and Wyner provided an analytic confirmation of Ungerboeck's observation by deriving a lower bound on the rates achievable with finite constellations \cite{ozarowwyner90}. In both works, the channel inputs are assumed to be uniformly distributed on a lattice within some enclosing boundary, where the size of the boundary is scaled in order to ensure unit input-power.

A related line of work considered signal constellations with favorable geometric properties, e.g., minimum Euclidean distance or minimum average error probability. For signal constellations with a large number of points, i.e., \emph{dense constellations}, Forney \emph{et al.} \cite{forneygallagerlanglongstaffqureshi84} estimated the loss in SNR with respect to the Gaussian input to be $10\log_{10} \frac{\pi e}{6} \approx 1.53$dB by comparing the volume of an $n$-dimensional hypercube with that of an $n$-dimensional hypersphere of identical average power. Later, Ungerboeck's work led to the study of multidimensional constellations based on lattices \cite{forneywei89}--\nocite{forney89,calderbankozarow90}\cite{kschischangpasupathy93}.

Recently, Wu and Verd\'u have studied the information rates that are achievable over the Gaussian channel when the input takes value in a finite constellation with $\const{N}$ signal points \cite{wuverdu10}. For every fixed SNR, they show that the difference between the capacity and the achievable rate tends to zero exponentially in $\const{N}$. For the optimal constellation, the peak-to-average-power ratio grows linearly with $\const{N}$, inducing no capacity loss. This is in contrast to the constellations considered by Ungerboeck \cite{ungerboeck82} and Ozarow and Wyner \cite{ozarowwyner90}, which have a finite peak-to-average-power ratio.

In this work, we adopt an information-theoretic perspective to study the capacity loss incurred by signal constellations with a bounded support over the Gaussian channel for sufficiently small noise variance. In particular, we use the duality-based upper bound to the mutual information in \cite{lapidothmoser03_3} to provide a lower bound on the capacity loss. The results are valid for both peak- and average-power constraints and generalize directly to other additive-noise channel models. For sufficiently high SNR, our results recover the power loss of $1.53$dB for square signal constellations without invoking geometrical arguments.

\section{Channel Model and Capacity}
\label{sec:channel}           

We consider a discrete-time, complex-valued additive noise channel, where the channel output $Y_k$ at time $k\in\Integers$ (where $\Integers$ denotes the set of integers) corresponding to the time-$k$ channel input $x_k$ is given by
\begin{equation}
\label{eq:channel}
Y_k = x_k + \sigma W_k, \quad k\in\Integers.
\end{equation}
We assume that $\{W_k,\,k\in\Integers\}$ is a sequence of independent and identically distributed, centered, unit-variance, complex random variables of finite differential entropy. We further assume that the distribution of $W_k$ does neither depend on $\sigma>0$ nor on the sequence of channel inputs $\{x_k,\,k\in\Integers\}$. 

The channel inputs take value in the set $\set{S}$, which is assumed to be a bounded Borel subset of the complex numbers $\Complex$. We further assume that $\set{S}$ has positive Lebesgue measure and that $0\in\set{S}$. 

The set $\set{S}$ can be viewed as the region that limits the signal points. For example, for a \emph{square signal constellation}, it is a square:
\begin{equation}
\label{eq:square}
\set{S}_{\sqbullet} \triangleq \{x\in\Complex\colon -\const{A}\leq\Re{x}\leq\const{A}, -\const{A}\leq\Im{x}\leq\const{A}\}
\end{equation}
for some $\const{A}>0$. Here $\Re{x}$ and $\Im{x}$ denote the real and imaginary part of $x$, respectively. Similarly, for a \emph{circular signal constellation},
\begin{equation}
\label{eq:circular}
\set{S}_{\bullet} \triangleq \{x\in\Complex\colon |x|\leq \const{R}\}, \quad \textnormal{for some $\const{R}>0.$}
\end{equation}

We study the capacity of the above channel under an average-power constraint $\const{P}$ on the inputs. Since the channel is memoryless, it follows that the capacity $C_{\set{S}}(\const{P},\sigma)$ (in nats per channel use) is given by
\begin{equation}
\label{eq:capacityS}
C_{\set{S}}(\const{P},\sigma) = \sup_{X\in\set{S}, \E{|X|^2}\leq\const{P}} I(X;Y)
\end{equation}
where the supremum is over all input distributions with essential support in $\set{S}$ that satisfy $\E{|X|^2}\leq\const{P}$.

We focus on $C_{\set{S}}(\const{P},\sigma)$ in the limit as the noise variance $\sigma$ tends to zero. In particular, we study the \emph{capacity loss}, which we define as
\begin{equation}
\label{eq:SH_def}
\const{L} \triangleq \lim_{\sigma\downarrow 0} \biggl\{C_{\Complex}(\const{P},\sigma)- C_{\set{S}}(\const{P},\sigma) \biggr\}.
\end{equation}
(Theorem~\ref{theorem} ahead asserts the existence of the limit.) Here $C_{\Complex}(\const{P},\sigma)$ denotes the capacity of the above channel when the support-constraint $\set{S}$ is relaxed, i.e.,
\begin{equation}
C_{\Complex}(\const{P},\sigma) = \sup_{\E{|X|^2}\leq\const{P}} I(X;Y).
\end{equation}
For small $\sigma$, we have \cite{shannon48}
\begin{equation}
C_{\Complex}(\const{P},\sigma) = \log\frac{\const{P}}{\sigma^2} + \log(\pi e) - h(W) + o(1)
\end{equation}
where the $o(1)$-term vanishes as $\sigma$ tends to zero. (Here $\log(\cdot)$ denotes the natural logarithm and $h(\cdot)$ denotes differential entropy.) The capacity loss \eqref{eq:SH_def} can thus be written as
\begin{IEEEeqnarray}{rCl}
\const{L}  &= & \log\const{P} + \log(\pi e) - h(W) \nonumber\\
& & {} -  \lim_{\sigma\downarrow 0} \biggl\{\sup_{X\in\set{S}, \E{|X|^2}\leq\const{P}} I(X;Y) - \log\frac{1}{\sigma^2} \biggr\} \label{eq:SH_expl}.
\end{IEEEeqnarray}

By choosing an input distribution that does not depend on $\sigma$, we can achieve\footnote{We define $h(X) = -\infty$ if the distribution of $X$ is not absolutely continuous with respect to the Lebesgue measure.}
\begin{equation}
\label{eq:SH_LB0.5}
\const{L} \leq \log\const{P} + \log(\pi e) - \sup_{X\in\set{S},\E{|X|^2}\leq\const{P}} h(X).
\end{equation}
Indeed, we have
\begin{IEEEeqnarray}{lCl}
I(X;Y) & = &  h(X+\sigma W) - h(W) + \log\frac{1}{\sigma^2} \label{eq:LB_1}
\end{IEEEeqnarray}
which follows from the behavior of differential entropy under deterministic translation and under scaling by a complex number. Extending \cite[Lemma 6.9]{lapidothmoser03_3} (see also \cite{linderzamir94}) to complex random variables yields then that, for every $\E{|X|^2}<\infty$ and $\E{|W|^2}<\infty$, the first differential entropy on the right-hand side (RHS) of \eqref{eq:LB_1} satisfies
\begin{equation}
\label{eq:h_cont}
\lim_{\sigma\downarrow 0} h(X+\sigma W) = h(X).
\end{equation}
Consequently, we obtain
\begin{IEEEeqnarray}{lCl}
\IEEEeqnarraymulticol{3}{l}{\lim_{\sigma\downarrow 0} \Biggl\{\sup_{X\in\set{S}, \E{|X|^2}\leq\const{P}} I(X;Y)-\log\frac{1}{\sigma^2}\Biggr\} }  \nonumber\\
\quad & \geq & \sup_{X\in\set{S}, \E{|X|^2}\leq\const{P}}\lim_{\sigma\downarrow 0} \Biggl\{ I(X;Y)-\log\frac{1}{\sigma^2}\Biggr\} \nonumber\\
& = & \sup_{X\in\set{S}, \E{|X|^2}\leq\const{P}} h(X) - h(W) \label{eq:SH_LB}
\end{IEEEeqnarray}
which together with \eqref{eq:SH_expl} yields \eqref{eq:SH_LB0.5}.

Let $\const{P}_{\set{U}}$ denote the average power of a random variable that is uniformly distributed over $\set{S}$, i.e.,
\begin{equation}
  \const{P}_{\set{U}} \triangleq \frac{\int_{\set{S}} |x|^2\d x}{\int_{\set{S}} \d x'}.
\end{equation}
A small modification of the proof in \cite[Th.~12.1.1]{coverthomas06} shows that the density that maximizes $h(X)$ for $X\in\set{S}$ with probability one and $\E{|X|^2}\leq\const{P}$ has the form
\begin{equation}\label{lemma2}
f_{\star}(x) = \frac{e^{-\lambda |x|^2}}{\int_{\set{S}} e^{-\lambda |x'|^2}\d x'}\I{x\in\set{S}}, \quad x\in\Complex
\end{equation}
where $\lambda=0$ for $\const{P}\geq\const{P}_{\set{U}}$, and where $\lambda$ satisfies 
\begin{equation}
\label{eq:lambda}
\frac{\int_{\set{S}}e^{-\lambda |x|^2}|x|^2 \d x}{\int_{\set{S}}e^{-\lambda |x'|^2}\d x'} = \const{P}
\end{equation}
for $\const{P}<\const{P}_{\set{U}}$. Here $\I{\textnormal{statement}}$ denotes the indicator function: it is equal to one if the statement in the brackets is true and it is otherwise equal to zero.

Applying \eqref{lemma2} to \eqref{eq:SH_LB0.5} yields
\begin{equation}
\label{eq:SH_LB_2}
\const{L} \leq \log\const{P} + \log(\pi e) - \log\biggl(\int_{\set{S}}e^{-\lambda |x'|^2}\d x'\biggr) - \lambda \,\const{P}.
\end{equation}
For $\const{P}=\const{P}_{\set{U}}$ (and hence $\lambda=0$), this becomes
\begin{equation}
\const{L} \leq \log(\pi e) + \log\biggl(\int_{\set{S}}|x|^2 \d x\biggr) - 2 \log\biggl(\int_{\set{S}} \d x\biggr). \label{eq:this}
\end{equation}
Specializing \eqref{eq:this} to a square signal constellation \eqref{eq:square} yields (irrespective of $\const{A}$)
\begin{equation}
\const{L}_{\sqbullet} \leq \log\frac{\pi e}{6}
\end{equation}
which corresponds to a power loss of roughly $1.53$dB. Hence, we recover the rule of thumb that ``square signal constellations have a $1.53$dB power loss at high signal-to-noise ratio."

For a circular signal constellation \eqref{eq:circular}, the upper bound \eqref{eq:this} becomes (irrespective of $\const{R}$)
\begin{equation}
\const{L}_{\bullet} \leq \log\frac{e}{2}
\end{equation}
recovering the power loss of $1.33$dB \cite{forneygallagerlanglongstaffqureshi84}.

The inequality in \eqref{eq:SH_LB_2} holds with equality if the capacity-achieving input-distribution does not depend on $\sigma$, cf.~\eqref{eq:SH_LB}. However, this is in general not the case. For example, for circularly-symmetric Gaussian noise and a circular signal constellation \eqref{eq:circular},
it was shown by Shamai and Bar-David \cite{shamaibardavid95} that, for every $\sigma>0$, the capacity-achieving input-distribution is discrete in magnitude, with the number of mass points growing with vanishing $\sigma$. Nevertheless, the following theorem demonstrates that the RHS of \eqref{eq:SH_LB_2} is indeed the capacity loss.

\begin{theorem}[Main Result]
\label{theorem}
For the above channel model, we have
\begin{equation}
\label{eq:SH}
\const{L} = \log\const{P} + \log(\pi e) - \log\biggl(\int_{\set{S}}e^{-\lambda |x'|^2}\d x'\biggr) - \lambda\, \const{P}
\end{equation}
where $\lambda=0$ for $\const{P}\geq\const{P}_{\set{U}}$, and where $\lambda$ satisfies \eqref{eq:lambda} for $\const{P}<\const{P}_{\set{U}}$.
\end{theorem}
\begin{proof}
See Section~\ref{sec:theorem}.
\end{proof}

\begin{note}
It is not difficult to adapt the proof of Theorem~\ref{theorem} to other regions $\set{S}$ and moment constraints. For example, the same proof technique can be used to derive the capacity loss when $\set{S}$ is a Borel subset of the real numbers and the channel input's first-moment is limited, i.e., $\E{|X|}\leq\const{A}$.
\end{note}

Equations \eqref{eq:LB_1}--\eqref{eq:SH_LB} demonstrate that the capacity loss \eqref{eq:SH} can be achieved with a continuous-valued channel input having density $f_{\star}(\cdot)$.  Using the lower-semicontinuity of relative entropy \cite{posner75}, it can be further shown that \eqref{eq:SH} can also be achieved by any sequence of discrete channel inputs $\{X_{\const{N}}\}$ for which the number of mass points $\const{N}$ grows with vanishing $\sigma$, provided that
\begin{equation}
X_{\const{N}} \stackrel{\mathscr{\scriptscriptstyle L}}{\to} X_{\star} \quad \textnormal{as $\const{N}\to\infty$}
\end{equation}
where $X_{\star}$ is a continuous random variable having density $f_{\star}(\cdot)$. (Here $\stackrel{\mathscr{\scriptscriptstyle L}}{\to}$ denotes convergence in distribution.) Such a sequence can, for example, be obtained by approximating the distribution function corresponding to $f_{\star}(\cdot)$ by two-dimensional step functions.

\section{Proof of Theorem~\ref{theorem}}
\label{sec:theorem}
In view of \eqref{eq:SH_expl}, in order to prove Theorem~\ref{theorem} it suffices to show that
\begin{IEEEeqnarray}{lCl}
 \IEEEeqnarraymulticol{3}{l}{\lim_{\sigma\downarrow 0} \biggl\{\sup_{X\in\set{S}, \E{|X|^2}\leq\const{P}} I(X;Y) - \log\frac{1}{\sigma^2} \biggr\}} \nonumber\\
\qquad & \leq &  \log\biggl(\int_{\set{S}}e^{-\lambda |x'|^2}\d x'\biggr) + \lambda\, \const{P} - h(W).\label{eq:proof_main}
\end{IEEEeqnarray}
The claim follows then by combining \eqref{eq:proof_main} with \eqref{eq:SH_LB_2}. To this end, we use the upper bound on the mutual information \cite[Th.~5.1]{lapidothmoser03_3}
\begin{equation}
\label{eq:proof_dual}
I(X;Y) \leq \int D\bigl(W(\cdot|x) \bigm\| R(\cdot)\bigr) \d Q(x)
\end{equation}
where $Q(\cdot)$ denotes the input distribution; $W(\cdot|x)$ denotes the conditional distribution of the channel output, conditioned on $X=x$; and $R(\cdot)$ denotes some arbitrary distribution on the output alphabet. Every choice of $R(\cdot)$ yields an upper bound on $I(X;Y)$, and the inequality in \eqref{eq:proof_dual} holds with equality if $R(\cdot)$ is the actual distribution of $Y$ induced by $Q(\cdot)$ and $W(\cdot|\cdot)$.

To derive an upper bound on $I(X;Y)$, we apply \eqref{eq:proof_dual} with $R(\cdot)$ having density
\begin{equation}
\label{eq:proof_pdf}
r(y) = \left\{\begin{array}{ll} \displaystyle \frac{e^{-\lambda |y|^2}}{\const{K}_{\eps,\sigma}}, \quad & \displaystyle y\in\set{S}_{\eps} \\[10pt] \displaystyle \frac{1}{\const{K}_{\eps,\sigma}}\frac{1}{\pi^2 \sigma |y|}\frac{1}{1+|y|/\sigma^2}, \quad & \displaystyle y\notin \set{S}_{\eps}\end{array}\right.
\end{equation}
where
\begin{equation}
\label{eq:proof_Keps}
\const{K}_{\eps,\sigma} \triangleq \int_{\set{S}_{\eps}} e^{-\lambda |y|^2} \d y + \int_{\cset{S}_{\eps}} \frac{1}{\pi^2\sigma |y|} \frac{1}{1+|y|^2/\sigma^2} \d y
\end{equation}
is a normalizing constant; where $\set{S}_{\eps}$ denotes the $\eps$-neighborhood of $\set{S}$
\begin{equation}
\set{S}_{\eps} \triangleq \bigl\{y\in\Complex\colon |y-x'|\leq \eps, \textnormal{for some $x'\in\set{S}$}\bigr\};
\end{equation}
where $\cset{S}_{\eps}$ denotes the complement of $\set{S}_{\eps}$; and where $\lambda$ is zero for $\const{P}\geq\const{P}_{\set{U}}$ and satisfies \eqref{eq:lambda} for $\const{P}<\const{P}_{\set{U}}$. Some useful properties of $\const{K}_{\eps,\sigma}$ are summarized in the following lemma.

\begin{lemma} \label{lemma3} \mbox{}
The normalizing constant $\const{K}_{\eps,\sigma}$ satisfies
\begin{subequations}
\begin{IEEEeqnarray}{rCl}
\inf_{\substack{\eps>0,\\\sigma>0}} \const{K}_{\eps,\sigma} & > & 0 \label{eq:lemma3_3} \\
\lim_{\eps\downarrow 0} \lim_{\sigma\downarrow 0} \const{K}_{\eps,\sigma} & = & \int_{\set{S}} e^{-\lambda |y|^2} \d y. \label{eq:lemma3_2}
\end{IEEEeqnarray}
\end{subequations}
\end{lemma}
\begin{proof}
Omitted.
\end{proof}

We return to the analysis of $I(X;Y)$ and apply \eqref{eq:proof_dual} together with the density \eqref{eq:proof_pdf} to express the upper bound as
\begin{IEEEeqnarray}{lCl}
\IEEEeqnarraymulticol{3}{l}{\int D\bigl(W(\cdot|x) \bigm\| R(\cdot)\bigr) \d Q(x)}\nonumber\\
\qquad & = & -h(Y|X) - \iint p(y|x) \log r(y) \d y \d Q(x) \label{eq:proof_1}
\end{IEEEeqnarray}
where $p(y|x)$ denotes the conditional probability density function of $Y$, conditioned on $X=x$.

Evaluation of the conditional differential entropy gives
\begin{equation}
\label{eq:h(Y|X)}
 h(Y|X) =   h(W) - \log\frac{1}{\sigma^2}
 \end{equation}
and some algebra applied to the second summand in \eqref{eq:proof_1} allows us to write it as
\begin{IEEEeqnarray}{lCl}
\IEEEeqnarraymulticol{3}{l}{-\iint p(y|x) \log r(y) \d y \d Q(x)}\nonumber\\
\quad\qquad & = &  \log\const{K}_{\eps,\sigma} + \lambda\, \E{|Y|^2 \I{Y\in\set{S}_{\eps}}} \nonumber\\
& & {} + \log\bigl(\pi^2\sigma^2\bigr) \Prob\bigl(Y\in\cset{S}_{\eps}\bigr) \nonumber\\
& & {} + \E{\log\biggl(\frac{|Y|}{\sigma}\biggr) \I{Y\in\cset{S}_{\eps}}} \nonumber\\ 
& & {}+ \E{\log\biggl(1+\frac{|Y|^2}{\sigma^2}\biggr) \I{Y\in\cset{S}_{\eps}}}. \label{eq:logR}
\end{IEEEeqnarray}
Combining \eqref{eq:h(Y|X)} and \eqref{eq:logR} with \eqref{eq:proof_1} and \eqref{eq:proof_dual} yields
\begin{IEEEeqnarray}{lCl}
\IEEEeqnarraymulticol{3}{l}{I(X;Y)}\nonumber\\
\quad & \leq & -h(W) + \log\frac{1}{\sigma^2} + \log\const{K}_{\eps,\sigma} + \lambda\, \E{|Y|^2 \I{Y\in\set{S}_{\eps}}}\nonumber\\
& & {} + \log\bigl(\pi^2\sigma^2\bigr) \Prob\bigl(Y\in\cset{S}_{\eps}\bigr) + \E{\log\biggl(\frac{|Y|}{\sigma}\biggr) \I{Y\in\cset{S}_{\eps}}} \nonumber\\ 
& & {}+ \E{\log\biggl(1+\frac{|Y|^2}{\sigma^2}\biggr) \I{Y\in\cset{S}_{\eps}}}. \label{eq:I(X;Y)UB}
\end{IEEEeqnarray}
We next show that, for $\eps>0$,
\begin{subequations}
\begin{IEEEeqnarray}{cCl}
\lim_{\sigma\downarrow 0} \sup_{X\in\set{S},\E{|X|^2}\leq \const{P}}\E{|Y|^2 \I{Y\in\set{S}_{\eps}}} & \leq & \const{P} \label{eq:proof_l1}\\
\lim_{\sigma\downarrow 0} \sup_{X\in\set{S},\E{|X|^2}\leq \const{P}}\Bigl|\log\bigl(\pi^2\sigma^2\bigr)\Prob\bigl(Y\in\cset{S}_{\eps}\bigr)\Bigr| & = & 0 \label{eq:proof_l2}\\
\lim_{\sigma\downarrow 0} \sup_{X\in\set{S},\E{|X|^2}\leq\const{P}} \Biggl|\E{\log\biggl(\frac{|Y|}{\sigma}\biggr) \I{Y\in\cset{S}_{\eps}}} \Biggr| & = & 0 \label{eq:proof_l3}\\
\lim_{\sigma\downarrow 0} \sup_{X\in\set{S},\E{|X|^2}\leq\const{P}} \E{\log\biggl(1+\frac{|Y|^2}{\sigma^2}\biggr) \I{Y\in\cset{S}_{\eps}}} & = & 0. \IEEEeqnarraynumspace\label{eq:proof_l4}
\end{IEEEeqnarray}
\end{subequations}

The first claim \eqref{eq:proof_l1} follows by upper-bounding
\begin{IEEEeqnarray}{lCl}
\IEEEeqnarraymulticol{3}{l}{\sup_{X\in\set{S},\E{|X|^2}\leq\const{P}} \E{|Y|^2 \I{Y\in\set{S}_{\eps}}}} \nonumber\\ 
\qquad & \leq & \sup_{X\in\set{S},\E{|X|^2}\leq\const{P}} \E{|Y|^2}\nonumber\\
& = & \sup_{X\in\set{S},\E{|X|^2}\leq\const{P}} \E{|X|^2} + \sigma^2 \E{|W|^2} \nonumber\\
& \leq & \const{P} + \sigma^2 \label{eq:proof_P+s}
\end{IEEEeqnarray}
where the second step follows because $X$ and $W$ are independent, and the third step follows because $\E{|X|^2}\leq\const{P}$ and $\E{|W|^2}=1$.

To prove \eqref{eq:proof_l2}, we first note that
\begin{equation}
\label{eq:35}
\Prob\bigl(Y\in\cset{S}_{\eps}\bigr) \leq \Prob\bigl(\sigma |W| > \eps\bigr).
\end{equation}
Indeed, if $|\sigma w|\leq\eps$, then we have $|y-x'|=|x+\sigma w-x'|\leq\eps$ for $x'=x\in\set{S}$, so $y\in\set{S}_{\eps}$. By Chebyshev's inequality \cite[Sec.~5.4]{gallager68}, this can be further upper-bounded by
\begin{IEEEeqnarray}{lCl}
\Prob\bigl(Y\in\cset{S}_{\eps}\bigr) & \leq & \frac{\sigma^2}{\eps^2}.\label{eq:proof_ch}
\end{IEEEeqnarray}
It then follows that, for $\sigma\leq\frac{1}{\pi}$,
\begin{equation}
0 \leq -\log\bigl(\pi^2\sigma^2\bigr)\Prob\bigl(Y\in\cset{S}_{\eps}\bigr) \leq -\log\bigl(\pi^2\sigma^2\bigr) \frac{\sigma^2}{\eps^2} \label{eq:36}
\end{equation}
where the right-most term vanishes as $\sigma$ tends to zero. This proves \eqref{eq:proof_l2}.

We next turn to \eqref{eq:proof_l3}. We first note that every $y\in\cset{S}_{\eps}$ must satisfy $|y|>\eps$, since otherwise $|y-x'|\leq\eps$ for $x'=0$, which by assumption is in $\set{S}$. Therefore,
\begin{IEEEeqnarray}{lCl}
\E{\log\biggl(\frac{|Y|}{\sigma}\biggr) \I{Y\in\cset{S}_{\eps}}} & \geq & \log\biggl(\frac{\eps}{\sigma}\biggr) \Prob\bigl(Y\in\cset{S}_{\eps}\bigr) \nonumber\\
& \geq & 0, \quad \textnormal{for $\sigma\leq\eps$.}\label{eq:proof_2}
\end{IEEEeqnarray}
To prove \eqref{eq:proof_l3}, it thus remains to show that
\begin{equation}
\label{eq:proof_3}
\lim_{\sigma\downarrow 0} \sup_{X\in\set{S},\E{|X|^2}} \E{\log\biggl(\frac{|Y|}{\sigma}\biggr) \I{Y\in\cset{S}_{\eps}}} \leq 0.
\end{equation}
By Jensen's inequality, we have
\begin{IEEEeqnarray}{lCl}
\IEEEeqnarraymulticol{3}{l}{\E{\log\biggl(\frac{|Y|}{\sigma}\biggr) \I{Y\in\cset{S}_{\eps}}}}\nonumber\\
\qquad & \leq & \Prob\bigl(Y\in\cset{S}_{\eps}\bigr) \log\Biggl(\frac{\E{|Y|\I{Y\in\cset{S}_{\eps}}}}{\sigma\Prob\bigl(Y\in\cset{S}_{\eps}\bigr)}\Biggr) \nonumber\\
& \leq & \frac{1}{2}\Prob\bigl(Y\in\cset{S}_{\eps}\bigr)\log\Biggl(\frac{\const{P}+\sigma^2}{\sigma^2\Prob\bigl(Y\in\cset{S}_{\eps}\bigr)}\Biggr) \label{eq:39}
\end{IEEEeqnarray}
where the last step follows from the Cauchy-Schwarz inequality
\begin{equation}
\E{|Y|\I{Y\in\cset{S}_{\eps}}} \leq \sqrt{\E{|Y|^2}\Prob\bigl(Y\in\cset{S}_{\eps}\bigr)}.
\end{equation}
Using \eqref{eq:proof_ch} together with the fact that $\xi\mapsto -\xi\log\xi$ is monotonically increasing for $\xi\leq e^{-1}$, we obtain for $\sigma\leq \eps\,e^{-1/2}$
\begin{IEEEeqnarray}{lCl}
\IEEEeqnarraymulticol{3}{l}{\E{\log\biggl(\frac{|Y|}{\sigma}\biggr) \I{Y\in\cset{S}_{\eps}}}}\nonumber\\
\qquad & \leq & \frac{1}{2} \frac{\sigma^2}{\eps^2} \log\biggl(1+\frac{\const{P}}{\sigma^2}\biggr) - \frac{\sigma^2}{2\eps^2} \log\frac{\sigma^2}{\eps^2} \label{eq:41}
\end{IEEEeqnarray}
from which \eqref{eq:proof_3}---and hence \eqref{eq:proof_l3}---follows by noting that the RHS of \eqref{eq:41} vanishes as $\sigma$ tends to zero.

To prove \eqref{eq:proof_l4}, we use Jensen's inequality and \eqref{eq:proof_P+s} to obtain
\begin{IEEEeqnarray}{lCl}
\IEEEeqnarraymulticol{3}{l}{\E{\log\biggl(1+\frac{|Y|^2}{\sigma^2}\biggr) \I{Y\in\cset{S}_{\eps}}}} \nonumber\\
\qquad & \leq & \Prob\bigl(Y\in\cset{S}_{\eps}\bigr) \log\Biggl(1+\frac{\E{|Y|^2\I{Y\in\cset{S}_{\eps}}}}{\sigma^2 \Prob\bigl(Y\in\cset{S}_{\eps}\bigr)}\Biggr) \nonumber\\
& \leq & \Prob\bigl(Y\in\cset{S}_{\eps}\bigr)\log\biggl(1+\frac{\const{P}}{\sigma^2} + \Prob\bigl(Y\in\cset{S}_{\eps}\bigr)\biggr) \nonumber\\
 & & {} - \Prob\bigl(Y\in\cset{S}_{\eps}\bigr)\log\Prob\bigl(Y\in\cset{S}_{\eps}\bigr). \label{eq:before_bla}
 \end{IEEEeqnarray}
Using \eqref{eq:proof_ch} together with the fact that $\xi\mapsto -\xi\log\xi$ is monotonically increasing for $\xi\leq e^{-1}$, we obtain for $\sigma\leq \eps\,e^{-1/2}$
 \begin{IEEEeqnarray}{lCl}
0 & \leq & \E{\log\biggl(1+\frac{|Y|^2}{\sigma^2}\biggr) \I{Y\in\cset{S}_{\eps}}}\nonumber\\ 
& \leq & \frac{\sigma^2}{\eps^2} \log\biggl(1+\frac{\const{P}}{\sigma^2} + \frac{\sigma^2}{\eps^2}\biggr) - \frac{\sigma^2}{\eps^2}\log\frac{\sigma^2}{\eps^2} \label{eq:bla}
\end{IEEEeqnarray}
from which \eqref{eq:proof_l4} follows by noting that the RHS of \eqref{eq:bla} vanishes as $\sigma$ tends to zero.

Combining \eqref{eq:proof_l1}--\eqref{eq:proof_l4} with \eqref{eq:I(X;Y)UB} yields
\begin{IEEEeqnarray}{lCl}
\IEEEeqnarraymulticol{3}{l}{\lim_{\sigma\downarrow 0} \Biggl\{\sup_{X\in\set{S},\E{|X|^2}\leq\const{P}} I(X;Y) - \log\frac{1}{\sigma^2}\Biggr\}}\nonumber\\
\qquad\qquad &\leq & {} -h(W) + \lim_{\sigma\downarrow 0} \log\const{K}_{\eps,\sigma} + \lambda\,\const{P} \nonumber\\
&= & {} -h(W) +  \log\biggl(\lim_{\sigma\downarrow 0}\const{K}_{\eps,\sigma}\biggr) + \lambda\,\const{P} \label{eq:proof_theorem_dobi}
\end{IEEEeqnarray}
where the last equation follows from the continuity of $x\mapsto \log(x)$ for $x>0$. Letting $\eps$ tend to zero, and using \eqref{eq:lemma3_2} in Lemma~\ref{lemma3}, we prove \eqref{eq:proof_main} and therefore the desired 
\begin{equation}
\const{L} = \log\const{P} + \log(\pi e) - \log\biggl(\int_{\set{S}} e^{-\lambda |y|^2}\d y\biggr) -\lambda\,\const{P}.
\end{equation}

\section{Nonasymptotic Capacity Loss}
A natural approach to prove Theorem~\ref{theorem} would be to generalize \eqref{eq:h_cont} to
\begin{equation}
\lim_{\sigma\downarrow 0} \sup_{X\in\set{S}, \E{|X|^2}\leq\const{P}} h(X+\sigma W) = \sup_{X\in\set{S}, \E{|X|^2}\leq\const{P}} h(X).
\end{equation}
While this approach may seem simpler, our approach has the advantage that it also allows for a lower bound on the \emph{nonasymptotic} capacity loss
\begin{equation}
\label{eq:loss_nonasymp}
\const{L}(\sigma) \triangleq C_{\Complex}(\const{P},\sigma)- C_{\set{S}}(\const{P},\sigma), \quad \sigma>0.
\end{equation}
Indeed, combining \eqref{eq:before_bla}, \eqref{eq:39}, and \eqref{eq:proof_P+s} with \eqref{eq:I(X;Y)UB} yields
\begin{IEEEeqnarray}{lCl}
I(X;Y) & \leq & -h(W) + \log\frac{1}{\sigma^2} + \log \const{K}_{\eps,\sigma} + \lambda \bigl(\const{P}+\sigma^2\bigr) \nonumber\\
& & {} + \log^+\bigl(\pi^2\sigma^2\bigr)\Prob\bigl(Y\in\cset{S}_{\eps}\bigr) \nonumber\\
& & {} + \frac{1}{2}\Prob\bigl(Y\in\cset{S}_{\eps}\bigr)\log\biggl(1+\frac{\const{P}}{\sigma^2}\biggr) \nonumber\\
& & {} +  \Prob\bigl(Y\in\cset{S}_{\eps}\bigr)\log\biggl(1+\frac{\const{P}}{\sigma^2} + \Prob\bigl(Y\in\cset{S}_{\eps}\bigr)\biggr) \nonumber\\
 & & {} - \frac{3}{2} \Prob\bigl(Y\in\cset{S}_{\eps}\bigr)\log\Prob\bigl(Y\in\cset{S}_{\eps}\bigr)
\end{IEEEeqnarray}
where $\log^+(\xi)\triangleq\max\{0,\log\xi\}$, $\xi>0$. By upper-bounding
\begin{equation}
\const{K}_{\eps,\sigma} \leq \int_{\set{S}_{\eps}} e^{-\lambda |y|^2} \d y + 1 - \frac{2}{\pi}\tan^{-1}\biggl(\frac{\eps}{\sigma}\biggr)
\end{equation}
(where $\tan^{-1}(\cdot)$ denotes the arctangent function), and by using \eqref{eq:35} together with the fact that $\xi\mapsto -\xi\log\xi$ is monotonically increasing for $\xi\leq e^{-1}$ and that $-\xi\log\xi \leq 1/e$ for $0<\xi<1$, we obtain, upon minimizing over $\eps>0$,
\begin{IEEEeqnarray}{lCl}
\IEEEeqnarraymulticol{3}{l}{C_{\set{S}}(\const{P},\sigma)} \nonumber\\
\quad & \leq & \inf_{\eps>0} \Biggl\{-h(W)+\log\frac{1}{\sigma^2} + \lambda\bigl(\const{P}+\sigma^2\bigr) \nonumber\\
& & {} +  \log\Biggl(\int_{\set{S}_{\eps}} e^{-\lambda |y|^2} \d y + 1 - \frac{2}{\pi}\tan^{-1}\biggl(\frac{\eps}{\sigma}\biggr) \Biggr) \nonumber\\
& & {} + \log^+\bigl(\pi^2\sigma^2\bigr)\Prob\bigl(\sigma |W| > \eps\bigr)\nonumber\\
& & {} + \frac{1}{2}\Prob\bigl(\sigma |W| > \eps\bigr)\log\biggl(1+\frac{\const{P}}{\sigma^2}\biggr) \nonumber\\
& & {} +  \Prob\bigl(\sigma |W| > \eps\bigr)\log\biggl(1+\frac{\const{P}}{\sigma^2} + \Prob\bigl(\sigma |W| > \eps\bigr)\biggr) \nonumber\\
& & {} - \frac{3}{2} \Prob\bigl(\sigma |W| > \eps\bigr)\log\Bigl(\Prob\bigl(\sigma |W| > \eps\bigr)\Bigr) \nonumber\\
& & \qquad\qquad\qquad\qquad\,\, {} \times \I{\Prob\bigl(\sigma |W| > \eps\bigr)\leq 1/e}\nonumber\\
& & \!\quad\qquad\qquad\qquad{} + \frac{3}{2e} \I{\Prob\bigl(\sigma |W| > \eps\bigr) > 1/e} \Biggr\}. \IEEEeqnarraynumspace
\end{IEEEeqnarray}
This together with \eqref{eq:loss_nonasymp} yields a lower bound on $\const{L}(\sigma)$.

\begin{figure}[t]
  \centering
  \psfrag{ylabel}[cb][cb]{$\const{L}(\sigma)$}
  \psfrag{xlabel}[ct][ct]{$1/\sigma^2$ [dB]}
  \psfrag{asymptotic}[c][c]{\footnotesize asymptotic capacity loss $\const{L}$}
  \psfrag{nonasymptotic}[lc][lc]{\footnotesize lower bound on $\const{L}(\sigma)$}
  \psfrag{1024}[rc][rc]{\tiny $2^{10}$-ary QAM}
  \psfrag{65536}[rc][rc]{\tiny $2^{16}$-ary QAM}
  \psfrag{4194304}[rc][rc]{\tiny $2^{22}$-ary QAM}
    \epsfig{file=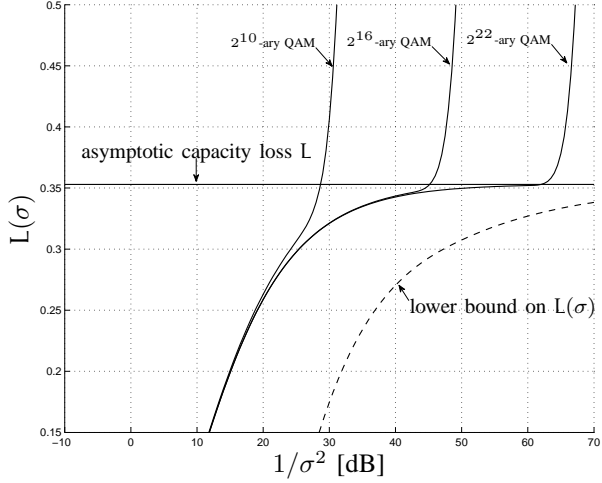, width=0.5\textwidth}
  \caption{The capacity loss $\const{L}(\sigma)$ for circularly-symmetric Gaussian noise and square constellations with $\const{P}=\const{P}_{\set{U}}$.}
  \label{fig}
\end{figure}

Figure~\ref{fig} shows the lower bound on $\const{L}(\sigma)$ for circularly-symmetric Gaussian noise and a  square signal constellation \eqref{eq:square} with $\const{P}=\const{P}_{\set{U}}$. It further shows the information-rate losses of $2^m$-ary quadrature amplitude modulation (QAM) for $m=10,16$, and $22$, which were numerically obtained using Gauss-Hermite quadratures \cite{abramowitzstegun72}, as described for example in \cite[Sec.~III]{alvaradobrannstromagrell11}. Since for a fixed $m$ the information rate corresponding to $2^m$-ary QAM is bounded by $m$ bits, the rate loss of $2^m$-ary QAM tends to infinity as $\sigma$ tends to zero. We observe that the lower bound on $\const{L}(\sigma)$ converges to \mbox{$\const{L}=\log(\pi e/6)\approx 0.353$} as $\sigma$ tends to zero, but is rather loose for finite $\sigma$. However, in the proof of Theorem~\ref{theorem} we chose the density \eqref{eq:proof_pdf} to decay sufficiently slowly, so as to ensure that the lower bound on $\const{L}$ holds for every unit-variance noise of finite differential entropy. For Gaussian noise, a density can be chosen that decays much faster, giving rise to a tighter bound.

\section*{Acknowledgment}
The authors would like to thank Alex Alvarado for helpful discussions and for providing the QAM curves in Figure~\ref{fig}.




\end{document}